\documentclass[conference]{IEEEtran}
\IEEEoverridecommandlockouts
\usepackage{cite}
\usepackage{amsmath,amssymb,amsfonts,bm}
\usepackage{amsthm}
\usepackage{acronym}
\usepackage[ruled, linesnumbered]{algorithm2e}
\usepackage{sansmath} 
\usepackage{upgreek}
\usepackage{graphicx}
\usepackage{subcaption}
\usepackage{enumitem}
\usepackage{pstricks}
\usepackage{textcomp}
\usepackage{xcolor}
\usepackage{microtype}
\newcommand{\RED}[1]{{\color[rgb]{1.00,0.10,0.10} #1}} 
 
\renewcommand{\RED}[1]{#1}

\usepackage{float} 

\long\def\comment#1{}

\newcommand{\PP}{{\mathbb P}}

\newcommand{\EE}{{\mathbb E}}

\newcommand{\Fc}{{\mathcal F}}

\newcommand{\ssf}{{\mathsf s}}

\newcommand{\xsf}{{\mathsf x}}

\newcommand{\Asf}{{\mathsf A}}

\newcommand{\Hsf}{{\mathsf H}}

\newcommand{\Qsf}{{\mathsf Q}}
\newcommand{\Rsf}{{\mathsf R}}

\newcommand{\Tsf}{{\mathsf T}}
\newcommand{\Usf}{{\mathsf U}}

\newcommand{\Xsf}{{\mathsf X}}
\newcommand{\Ysf}{{\mathsf Y}}
\newcommand{\Zsf}{{\mathsf Z}}

\renewcommand{\det}{{\hbox{det}}}

\renewcommand{\arg}{{\hbox{arg}}}

\newcommand{\SNR}{{\mathsf {SNR}}}

\newcommand{\herm}{{\mathsf H}}
\newcommand{\trasp}{{\mathsf T}}

\def\E{{\mathbb E}}

\def\tr{{\rm tr}}

\newtheorem{thm}{Theorem}
\newtheorem{cor}{Corollary}
\newtheorem{lem}{Lemma}
\newtheorem{prop}{Proposition}
\newtheorem{rem}{Remark}

\newtheorem{defn}{\protect\definitionname}

\providecommand{\definitionname}{Definition}

\acrodef{aoa}[AOA]{angle-of-arrival}
\acrodef{bcrb}[BCRB]{Bayesian Cram\'{e}r-Rao bound}
\acrodef{bfim}[BFIM]{Bayesian Fisher information matrix}
\acrodef{bp}[BP]{belief propagation}
\acrodef{cdi}[CDI]{cooperative dilution intensity}
\acrodef{cir}[CIR]{channel impulse response}
\acrodef{cl}[CL]{cooperative localization}
\acrodef{cp}[CP]{cyclic prefix}
\acrodef{crb}[CRB]{Cram\'{e}r-Rao bound}
\acrodef{crlb}[CRLB]{Cram\'{e}r-Rao lower bound}
\acrodef{dft}[DFT]{discrete Fourier transform}
\acrodef{dof}[DoF]{degree of freedom}
\acrodef{dpeb}[DPEB]{directional position error bound}
\acrodef{fim}[FIM]{Fisher information matrix}
\acrodef{efim}[EFIM]{equivalent Fisher information matrix}
\acrodef{hogs}[HOGS]{hybrid orthogonal-Gaussian signalling}
\acrodef{hsug}[HSUG]{hybrid semi-unitary--Gaussian}
\acrodef{ici}[ICI]{information coupling intensity}
\acrodef{icrb}[ICRB]{inverse CRB}
\acrodef{iid}[i.i.d.]{independently and identically distributed}
\acrodef{im}[IM]{index modulation}
\acrodef{isac}[ISAC]{Integrated Sensing and Communication}
\acrodef{los}[LoS]{line-of-sight}
\acrodef{mse}[MSE]{mean-squared error}
\acrodef{ofdm}[OFDM]{orthogonal frequency-division multiplexing}
\acrodef{pdf}[PDF]{probability density function}
\acrodef{peb}[PEB]{position error bound}
\acrodef{speb}[SPEB]{squared position error bound}
\acrodef{pll}[PLL]{phase-locked loop}
\acrodef{psk}[PSK]{phase shift keying}
\acrodef{p2p}[P2P]{point-to-point}
\acrodef{qam}[QAM]{quadrature amplitude modulation}
\acrodef{rbs}[RBS]{reference broadcast synchronization}
\acrodef{rhs}[RHS]{right hand side}
\acrodef{rii}[RII]{ranging information intensity}
\acrodef{rss}[RSS]{received signal strength}
\acrodef{rc}[RC]{ranging coefficient}
\acrodef{speb}[SPEB]{squared position error bound}
\acrodef{toa}[TOA]{time-of-arrival}
\acrodef{tdoa}[TDOA]{time-difference-of-arrival}
\acrodef{tpsn}[TPSN]{time synchronization protocol for sensor network}
\acrodef{vmp}[VMP]{variational message passing}
\acrodef{wsn}[WSN]{wireless sensor network}
\acrodef{efim}[EFIM]{equivalent Fisher information matrix}
\acrodef{dio}[DIO]{distance-information-only}
\acrodef{aio}[AIO]{angle-information-only}
\acrodef{saaf}[SAAF]{squared array aperture function}
\acrodef{snc}[S\&C]{sensing and communications}
\acrodef{uoa}[UOA]{uniformly oriented array}
\acrodef{rgg}[RGG]{random geometric graph}
\acrodef{snr}[SNR]{signal-to-noise ratio}
\acrodef{eoc}[EoC]{efficiency of cooperation}
\acrodef{npi}[NPI]{nominal position information}
\acrodef{gnss}[GNSS]{global navigation satellite system}
\acrodef{mimo}[MIMO]{multiple-input multiple-output}
\acrodef{mcs}[MCS]{minimally constrained system}
\acrodef{zzb}[ZZB]{Ziv-Zakai lower bound}
\acrodef{wwb}[WWB]{Weiss-Weinstein lower bound}
\acrodef{nlos}[NLOS]{non-light-of-sight}
\acrodef{mmse}[MMSE]{minimum mean squared error}
\acrodef{emmse}[EMMSE]{ergodic minimum mean squared error}

\acrodef{lmmse}[LMMSE]{linear minimum mean squared error}
\acrodef{uav}[UAV]{unmanned aerial vehicle}
\acrodef{ppp}[PPP]{Poisson point process}
\acrodef{bpp}[BPP]{binomial point process}
\acrodef{cln}[CLN]{cooperative location-aware network}
\acrodef{pdr}[PDR]{pedestrian dead reckoning}
\acrodef{ml}[ML]{maximum likelihood}
\acrodef{map}[MAP]{maximum \textit{a posteriori}}
\acrodef{kkt}[KKT]{Karush-Kuhn-Tucker}
\acrodef{st}[ST]{subspace tradeoff}
\acrodef{drt}[DRT]{deterministic-random tradeoff}
\acrodef{ustm}[USTM]{unitary space-time modulation}
\acrodef{tir}[TIR]{target impulse response}
\acrodef{pga}[PGA]{projected gradient algorithm}

\acrodefplural{dof}[DoFs]{degrees of freedom}
\acrodefplural{drt}[DRTs]{deterministic-random tradeoffs}

\acrodef{id}[i.d.]{isotropically distributed}

\acrodef{csi}[CSI]{channel state information}
  
\acrodef{mi}[MI]{mutual information}
\def\EMMSE{\mbox{\scriptsize\sf EMMSE}}
\def\BibTeX{{\rm B\kern-.05em{\sc i\kern-.025em b}\kern-.08em
    T\kern-.1667em\lower.7ex\hbox{E}\kern-.125emX}}
\begin{document}
\setlength{\abovedisplayskip}{3pt}
\setlength{\belowdisplayskip}{3pt}
\setlength{\abovedisplayshortskip}{1pt}
\setlength{\belowdisplayshortskip}{1pt}
\title{Noncoherent ISAC over Block-Fading Channels: Asymptotic Performance Analysis}

\author{
\IEEEauthorblockN{%
Hao~Yang\IEEEauthorrefmark{1},
Kai~Wan\IEEEauthorrefmark{1},
Giuseppe~Caire\IEEEauthorrefmark{2}
}
\IEEEauthorblockA{\IEEEauthorrefmark{1}Huazhong University of Science and Technology, 430074  Wuhan, China, \{hao\_yang,kai\_wan\}@hust.edu.cn}%
\IEEEauthorblockA{\IEEEauthorrefmark{2}Technische Universit\"at Berlin, 10587 Berlin, Germany, caire@tu-berlin.de}%
}

\maketitle

\begin{abstract}
This paper investigates the fundamental limits and optimal signal distribution design for Integrated Sensing and Communication (ISAC) systems operating under strictly noncoherent conditions. Unlike conventional coherent frameworks that rely on perfect channel state information, we consider a block-fading MIMO channel where the channel realizations are unknown to both the transmitter and the receiver. We adopt a realization-wise perspective to characterize the noncoherent performance tradeoff across different signal-to-noise ratio (SNR) regimes. In the high-SNR regime, we derive a lower bound for the noncoherent mutual information and define a metric, termed sensing-induced rate loss, to quantify the communication penalty incurred by sensing-oriented beamforming. We then employ a projected gradient algorithm to optimize the spatial power allocation, balancing the conflict between the unitary space-time modulation-based structure for communication and the task-oriented spatial power allocation for sensing. Conversely, in the low-SNR regime, we perform a first-order asymptotic analysis of the ergodic minimum mean squared error (EMMSE). Our theoretical derivation reveals a fundamental synergy: the sensing-optimal strategy collapses to a rank-one transmission along the dominant eigenvector of the target response, which incurs no first-order communication loss in the low-SNR regime. This result demonstrates that the conflicting tradeoff observed at high SNR vanishes asymptotically at low SNR, enabling perfect alignment between sensing and communication objectives.
\end{abstract}


\acresetall
\section{Introduction}

The paradigm of \ac{isac} has emerged as a transformative architecture for future wireless networks, predicated on the synergistic integration of hardware and spectral resources. To date, the majority of information-theoretic milestones in \ac{isac} have been established under the assumption of a coherent regime, where perfect \ac{csi} is presumed to be available at both the transmitter and receiver~\cite{xiongFundamentalTradeoffIntegrated2023, ahmadipourInformationtheoreticApproachJoint2024, yilmazJointCommunicationParameter2026}. Within this coherent framework, extensive research has been devoted to characterizing the Pareto boundary between communication mutual information and sensing performance, quantified via \ac{bcrb}-based metrics (e.g.,~\cite{xiongFundamentalTradeoffIntegrated2023, guoFundamentalLimitsISAC2025}) or \ac{mmse}-type metrics~\cite{wangFundamentalMMSEratePerformance2025,shenFundamentalTradeoffBistatic2025}, while alternative formulations based on sensing mutual information have also been proposed~\cite{yuRethinkingFundamentalPerformance2025}. However, these idealized gains often hinge on the implicit assumption that the resources dedicated to channel training and feedback are negligible---a premise that inevitably collapses in highly dynamic or power-constrained environments. Consequently, investigating \ac{isac} systems under strictly noncoherent conditions---where instantaneous channel realizations remain unknown to all nodes---is of paramount importance for characterizing the fundamental limits of practical deployments.

Building upon these coherent results, recent literature has shifted toward optimizing \ac{isac} systems under more realistic \ac{csi} configurations and signaling constraints. For instance, the authors in~\cite{kobayashiJointStateSensing2019} showed that in a state-dependent multi-access channel (MAC) with generalized feedback, the lack of instantaneous \ac{csi} at the transmitters requires carefully designed signaling strategies, even when the receiver observes the channel outputs perfectly. Furthermore, investigations into signal shaping~\cite{liuDeterministicrandomTradeoffIntegrated2023, duReshapingISACTradeoff2024} have exposed a critical divergence between deterministic and random signaling: while Gaussian stochasticity is conventionally optimal for communication capacity, its inherent uncertainty may compromise sensing resolution. This tension necessitates a judicious balance in the design of the underlying signal structure.

At a more fundamental level, the information-theoretic framework for noncoherent joint transmission and state sensing was pioneered by Zhang et {\it al.}~\cite{zhangJointTransmissionState2011}, who conceptualized the sensing task through the lens of a state-dependent memoryless channel. By assuming the state is unknown to both terminal nodes, they characterized the tradeoff between the achievable communication rate and state estimation distortion. This framework was subsequently refined by Choudhuri {\it et al.}~\cite{choudhuriCausalStateCommunication2013}, who incorporated causal and strictly causal state information at the transmitter, thereby establishing the analytical basis for evaluating how asymmetric environmental knowledge dictates system performance bounds. Related formulations consider fixed but unknown channel states and use a detection-error exponent as the sensing metric; see, e.g., Chang et {\it al.}~\cite{changRateDetectionerrorExponent2023}, who characterized the rate--detection-exponent tradeoff for joint communication and sensing of discrete channel states.

Despite these advances, the optimal design for strictly noncoherent \ac{isac} remains under-explored, particularly regarding the influence of the \ac{snr} regime. In the high-\ac{snr} regime, the noncoherent channel is primarily degree of freedom limited; here, enforcing a sensing-optimal covariance structure creates a structural mismatch with the capacity-achieving \ac{ustm} profile characterized by~\cite{hochwald_unitary_2000, lizhong_zheng_communication_2002}, leading to a quantifiable capacity penalty. In stark contrast, the low-\ac{snr} regime is power-limited \cite{verdu_spectral_2002}, where the stringent requirements for spatially uniform signaling relax as the capacity scales linearly with the total received power. This observation raises the question of whether the \ac{snc} tradeoff persists across all power regimes. Specifically, it remains to be determined whether the structural conflict observed at high \ac{snr} is fundamental, or if the two objectives asymptotically align in the low-\ac{snr} regime, enabling synergistic operation.

Motivated by these considerations, this paper investigates the optimization of \ac{isac} systems operating under strictly noncoherent conditions. We consider a block-fading MIMO Gaussian channel, where the channel matrix remains constant over a coherence interval of $T$ symbols but its realization is unknown to all nodes. Our contributions are three-fold:

    $\bullet$ \textit{Optimal signaling structure:} We derive the fundamental signal structure for noncoherent \ac{isac} systems. We prove that the optimal sensing strategy---a deterministic sample covariance aligned with the target's eigen-structure---admits an \textit{effective communication input} representation. By exploiting the unitary invariance of isotropic block-fading channels, we demonstrate that the spatial rotation required for sensing is transparent to the communication receiver. This structural insight rigorously reduces the joint design from complex matrix optimization to  tractable power allocation  over spatial eigenmodes.

    $\bullet$ \textit{High-SNR tradeoff and rate loss quantification:} We derive a high-SNR lower bound on the noncoherent mutual information for block-fading MIMO channels. Based on this bound, we introduce the \textit{sensing-induced rate loss} to quantify the penalty incurred by deviating from equal-power signaling (i.e., USTM) to satisfy sensing-specific requirements. Consequently, we formulate a convex optimization problem for spatial power allocation to trace the system's Pareto frontier.
    
    $\bullet$ \textit{Low-SNR asymptotic alignment:} We prove analytically that the sensing-optimal strategy collapses to a rank-one beamforming along the dominant eigenmode of the target. Crucially, our analysis reveals that this highly structured transmission incurs no penalty in the first-order communication rate, as the noncoherent capacity becomes insensitive to the spatial covariance structure in this regime. This highlights an interesting phenomenon: the conflicting tradeoff observed at high SNR vanishes asymptotically at low SNR, allowing for perfect alignment of \ac{snc}  objectives.

\textit{Notations:}
$\xsf$, $\bm{\xsf}$, and $\bm{\Xsf}$ represent random variables, random vectors, and random matrices, respectively, with their deterministic counterparts denoted by $x$, $\bm{x}$, and $\bm{X}$. 
The $M \times M$ identity matrix is denoted by $\bm{I}_{M}$. 
$\E[\cdot]$ and $\PP[\cdot]$ denote the expectation and probability operators, respectively. 
The transpose and Hermitian transpose are denoted by $[\cdot]^{\Tsf}$ and $[\cdot]^{\mathsf{H}}$, respectively. 
$[\bm{X}]_{ij}$ denotes the $(i,j)$-th entry of matrix $\bm{X}$. 
$\mathrm{diag}(\bm{x})$ denotes a diagonal matrix with the elements of $\bm{x}$ on its main diagonal, and $\mathrm{tr}(\cdot)$ denotes the trace of a matrix.
$\bm{A} \succeq \bm{B}$ indicates that $\bm{A}-\bm{B}$ is positive semidefinite. 

\section{System Model and Preliminaries}
\label{sec:model}
\subsection{System Model}
Let us consider the general \ac{isac} system model given by
\begin{align}
\bm{\Ysf}_i = \bm{\Hsf}_i \bm{\Xsf} + \bm{\Zsf}_i,\quad i\in\{\mathrm{c},\mathrm{s}\},
\end{align}
where $\bm{\Hsf}_{\rm c}\in\mathbb{C}^{N_{\rm c}\times M}$ and $\bm{\Hsf}_{\rm s}\in\mathbb{C}^{N_{\rm s}\times M}$ denote the communication channel and the target response matrix, respectively, \RED{with $M$, $N_{\rm c}$, and $N_{\rm s}$ being the number of antennas at the \ac{isac} transmitter (Tx), the communication receiver (Rx1), and the sensing receiver (Rx2), respectively. $\bm{\Ysf}_{\rm c}\in\mathbb{C}^{N_{\rm c}\times T}$ and $\bm{\Ysf}_{\rm s}\in\mathbb{C}^{N_{\rm s}\times T}$ denote the received communication and sensing signals over a coherence block of length $T$, respectively.} $\bm{\Xsf}\in\mathbb{C}^{M\times T}$ denotes the transmitted dual-functional waveform for performing both \ac{snc} tasks. In this paper, we assume that $\bm{\Xsf}$ is known to both the Tx and Rx2 \RED{(e.g., as is typical in monostatic sensing setups or bistatic systems with perfect backhaul),} but is unknown to Rx1. 
Neither the Tx nor Rx1 has access to realizations of $\bm{\Hsf}_{\rm c}$, and each element of $\bm{\Hsf}_{\rm c}$, $h_{ij}$ for $i = 1, \dots, N_{\rm c}$ and $j = 1, \dots, M$, is assumed to be \ac{iid} zero-mean circularly symmetric complex Gaussian with unit variance, i.e., $h_{ij} \sim \mathcal{CN}(0,1)$. The noise matrices $\bm{\Zsf}_{\rm c}$ and $\bm{\Zsf}_{\rm s}$ have \ac{iid} $\mathcal{CN}(0,\sigma_{\rm c}^2)$ and $\mathcal{CN}(0,\sigma_{\rm s}^2)$ entries, respectively.
The communication subsystem aims at transmitting as much information as possible (reliably) to Rx1, while the sensing subsystem aims at estimating $\bm{\Hsf}_{\rm s}$ at Rx2.

We consider a block-fading model for both the sensing channel and the communication channel. Specifically, we assume that $\bm{\Hsf}_{\rm s}$ varies every $T$ symbols in an \ac{iid} manner, following a statistical correlation matrix \(\bm{R}_{{\bm \Hsf}_{\rm s}} = \EE[\bm \Hsf_{\rm s}^{\herm}{\bm \Hsf}_{\rm s}]\), and that the communication channel $\bm{\Hsf}_{\rm c}$ also varies every $T$ symbols in an \ac{iid} manner. We will refer to $T$ as the \emph{coherence time}\footnote{This identical coherence time assumption is standard in information-theoretic ISAC studies (cf.~\cite{xiongFundamentalTradeoffIntegrated2023, ahmadipourInformationtheoreticApproachJoint2024, yilmazJointCommunicationParameter2026}) to model scenarios where targets and users share similar mobilities. Extending to dual-timescale models ($T_{\rm c} \neq T_{\rm s}$) breaks the required unitary invariance and is left for future work.} in the rest of the paper. In this strictly noncoherent setting, no terminal has access to the instantaneous realizations of $\bm{\Hsf}_{\rm c}$ or $\bm{\Hsf}_{\rm s}$; only their distributions and, for sensing, the second-order statistics $\bm{R}_{{\bm \Hsf}_{\rm s}}$ are assumed known.

 The average transmit power at each transmit antenna in one symbol period is normalized to $1$, and thus the power constraint can be written as
\begin{equation}\label{power_constraint}
 \E\bigl[\tr({\bm{\Xsf\Xsf^\herm}})\bigr] = MT.
\end{equation}
We refer to the \ac{snr} as the average \ac{snr} at each receive antenna. Under the aforementioned normalization, the communication \ac{snr} and sensing \ac{snr} are respectively defined as
\begin{equation}
    \SNR_{\rm c} \triangleq \frac{M}{\sigma_{\rm c}^2} \quad \text{and} \quad {\SNR}_{\rm s} \triangleq \frac{M}{\sigma_{\rm s}^2}.
\end{equation}

The performance of the communication subsystem is conventionally characterized by 
\begin{equation}
C = \sup_{p_{\bm{\Xsf}}(\bm{X})} ~T^{-1} I(\bm{\Ysf}_{\rm c};\bm{\Xsf}), \quad {\rm s.t.}~p_{\bm{\Xsf}}(\bm{X})\in\Fc,
\end{equation}
where $I(\bm{\Ysf}_{\rm c};\bm{\Xsf})$ denotes the mutual information between $\bm{\Ysf}_{\rm c}$ and $\bm{\Xsf}$, and $\Fc$ denotes the feasibility region of $p_{\bm{\Xsf}}(\bm{X})$ determined by  the power constraint \eqref{power_constraint}.

We utilize the ergodic \ac{mmse} (EMMSE)\footnote{We adopt EMMSE because it leads to a more tractable optimization framework and remains a tighter performance indicator in the low-SNR noncoherent regime, where derivative-based bounds like the \ac{bcrb} often become overly optimistic or difficult to compute.} to quantify sensing performance \cite{luRandomISACSignals2024}. This metric is formulated based on the linear \ac{mmse} estimator, under the assumption that channel statistics are known.
The ergodic estimation error can be computed as
\begin{subequations}\label{def:emmse}
\begin{align}
J_{\EMMSE}&\triangleq \E_{\bm{\Xsf}}\left[\tr \left( (\bm{R}_{\bm{\Hsf}_{\rm s}}^{-1}+\kappa\bm{\Xsf} \bm{\Xsf}^\herm)^{-1}\right) \right] \\
&= \E_{\bm{\Rsf}_{\bm{\Xsf}}}\left[\tr \left( (\bm{R}_{\bm{\Hsf}_{\rm s}}^{-1}+\kappa T\bm{\Rsf}_{\bm{\Xsf}})^{-1} \right) \right], 
\end{align}
\end{subequations}
where $\kappa={1}/{\sigma_{{\rm s}}^2 N_{\rm s} }={\SNR_{\rm s}}/M{N_{\rm s}}$ and $\bm{\Rsf}_{\bm{\Xsf}}=T^{-1}\bm{\Xsf} \bm{\Xsf}^\herm$ is the sample covariance matrix, and $\bm{R}_{\bm{\Hsf}_{\rm s}}$ is assumed to be full-rank, so that $\bm{R}_{\bm{\Hsf}_{\rm s}}^{-1}$ exists.
\subsection{Preliminary Results on Noncoherent Communication }
For noncoherent communication models where the fading coefficients are unavailable at both the receiver and the transmitter (i.e., by removing the sensing task from our considered problem), the channel capacity has been characterized in \cite{marzettaCapacityMobileMultipleantenna1999,lizhong_zheng_communication_2002}. We recall the main structural results as follows.

\begin{defn}[\cite{lizhong_zheng_communication_2002}]
A random matrix $\bm{\Rsf} \in \mathbb{C}^{M \times T}$, for $T \geq M$, is called \ac{id} if 
\(
p(\bm{\Rsf}) = p(\bm{\Rsf} \bm{U})
\)
for any deterministic $T \times T$ unitary matrix $\bm{U}$.
\end{defn}

\begin{lem}[\cite{lizhong_zheng_communication_2002}]
\label{id}
Let $\bm{\Hsf}$ be \ac{id}, and let $\bm{\Usf}$ be a random unitary matrix independent of $\bm{\Hsf}$. Then, for any realization $\bm{\Usf} = \bm{U}$, the distribution of $\bm{\Hsf}\bm{U}$ is identical to that of $\bm{\Hsf}$. Accordingly, $\bm{\Hsf}\bm{\Usf}$ and $\bm{\Usf}$ are statistically independent.
\end{lem}

\begin{lem}[\cite{marzettaCapacityMobileMultipleantenna1999}]
\label{AQ}
The input distribution that achieves capacity can be written as 
\(
\bm{\Xsf} = \bm{\Asf} \bm{\Qsf},
\)
where $\bm{\Qsf}$ is an $M \times T$ \ac{id} semi-unitary matrix, i.e., $\bm{\Qsf} \bm{\Qsf}^{\herm} = \bm{I}_M$, and $\bm{\Asf}$ is an $M \times M$ real diagonal random matrix. Moreover, $\bm{\Asf}$ and $\bm{\Qsf}$ are independent.
\end{lem}

This characterization reduces the optimization dimensionality from $MT$ to $M$. 
At high SNR, the equal-power input \(\PP \left[\bm{\Asf} = \sqrt{T} \, \bm{I}_M \right] = 1\) is asymptotically optimal, yielding explicit expressions for the high-SNR capacity derived in \cite{lizhong_zheng_communication_2002}.

\begin{lem}[\cite{lizhong_zheng_communication_2002}]
\label{cpure}
For the noncoherent MIMO channel with $M \le N_{\rm c}$ and $T \ge N_{\rm c}+M$, the high-SNR capacity is
\begin{equation}
    C_{\rm pure} = M\left(1 - \frac{M}{T}\right)\log_2 \SNR_{\rm c} + c + o(1),
\end{equation}
where
\begin{equation*}
    \begin{split}
        c =& \frac{1}{T} \log_2 |G(T,M)| + M\left(1 - \frac{M}{T}\right)\log_2 \frac{T}{M\pi e} \\
          &\ + \left(1 - \frac{M}{T}\right) \E[\log_2 \operatorname{det}\bm{\Hsf}_{\rm c}^\herm\bm{\Hsf}_{\rm c}],
    \end{split}
\end{equation*}
with \RED{$|G(T,M)|$ denoting the volume of the complex Grassmann manifold $G(T,M)$.}
\end{lem}
\section{Sensing-Optimal Design and Effective Communication Input}
\label{sec:sopt}
In this section, we establish the fundamental signaling structure for the noncoherent ISAC system. We first derive the sensing-optimal transmit strategy and then connect it to the noncoherent communication framework. Exploiting the unitary invariance of the isotropic communication channel, we show that the spatial rotation required for sensing is transparent to communication, so that the joint design reduces to a tractable optimization over spatial power allocation.

\subsection{Sensing-Optimal Design}
\begin{prop} \label{sensing_optimal}
To minimize the $J_{\EMMSE}$ under the power constraint \eqref{power_constraint}, the optimal transmit signal $\bm{\Xsf}^{\star}$ must satisfy:

    1) The sample covariance matrix ${\bm{\Rsf}}_{\bm{\Xsf}}$ is deterministic.
    
    2) The optimal sample covariance follows the structure
      \begin{equation}
    \bm{\Rsf}_{\bm{\Xsf}}^{\star} = \bm{U}_{\rm s} (\bm{D}_{\rm s}^{\star})^{2} \bm{U}_{\rm s}^{\herm},
       \end{equation}
    where $\bm{U}_{\rm s}$ is the \RED{unitary} eigenvector matrix of $\bm{R}_{\bm{\Hsf}_{\rm s}} = \bm{U}_{\rm s} \bm{\varLambda}_{\rm s} \bm{U}_{\rm s}^\herm$.
    
    3) \textit{Water-filling Power Allocation:} The $i$-th diagonal element of $(\bm{D}_{\rm s}^{\star})^{2}$, denoted by $(d_{i}^{\star})^{2}$, $d_{i}\geq0$, is given by
    \begin{equation} \label{eq:wf_solution}
    (d_{i}^{\star})^{2} = \frac{1}{\kappa T} \left({{\nu}} - \frac{1}{\lambda_{i}} \right)^+, \, \forall i \in \{1, \dots, M\},
    \end{equation}
    where $\lambda_{i}$ is the $i$-th eigenvalue of $\bm{R}_{\bm{\Hsf}_{\rm s}}$, and $\nu \geq 0$ is the Lagrange multiplier ensuring $\sum_{i=1}^M (d_{i}^{\star})^{2} = M$.
    
    4) If $\bm{\Hsf}_{\rm s}$ is \ac{id}, then $\bm{R}_{\bm{\Hsf}_{\rm s}} \propto \bm{I}_M$ and $d_{i}^{\star} = 1, \forall i$.
\end{prop}
\begin{proof}
(sketch) Convexity of $J_{\EMMSE}$ with respect to ${\bm{\Rsf}}_{\bm{\Xsf}}$ and Jensen's inequality imply that the optimal $\bm{\Rsf}_{\bm{\Xsf}}$ is deterministic. Eigenvector alignment and the water-filling structure follow from standard trace inequalities and KKT conditions.
For a detailed proof, please refer to Appendix~\ref{app:sensing_optimal}.
\end{proof}

\subsection{Effective Communication Input}\label{subsec:eff_comm_input}
\RED{Note that Proposition~\ref{sensing_optimal} only imposes a constraint on the sample covariance matrix, implying that the sensing-optimal transmit signal is not unique. To simultaneously fulfill this sensing requirement and the noncoherent capacity-achieving structure in Lemma~\ref{AQ}, a structurally compatible transmit signal is constructed as} $\bm{\Xsf}^{\star} = \sqrt{T}\,\bm{U}_{\rm s} \bm{D}_{\rm s}^{\star} \bm{\Qsf},$
where $\bm{D}_{\rm s}^{\star}$ is defined in Proposition~\ref{sensing_optimal} and $\bm{\Qsf}$ is Haar-distributed semi-unitary.
In isotropic scattering environments, the communication channel $\bm{\Hsf}_{\rm c}$ is \ac{id}. By Lemma~\ref{id}, we have
$$
\bm{\Hsf}_{\rm c} \bm{\Xsf}^{\star}
= \sqrt{T} (\bm{\Hsf}_{\rm c} \bm{U}_{\rm s}) \bm{D}_{\rm s}^{\star} \bm{\Qsf}
\overset{{\rm d}}{=} \sqrt{T} \bm{\Hsf}_{\rm c} \bm{D}_{\rm s}^{\star} \bm{\Qsf},
$$
where $\overset{\rm d}{=}$ denotes equality in distribution.
This equivalence indicates that the unitary precoder $\bm{U}_{\rm s}$ designed for sensing eigenspace alignment is completely absorbed into the channel distribution and \RED{thus does not affect the ergodic communication rate.}
Consequently, from the communication perspective, the effective input under the sensing-optimal design is given by
\begin{equation}
\label{xeff}
\bm{\Xsf}_{\text{eff}}^{\star} = \sqrt{T}\,\bm{D}_{\rm s}^{\star} \bm{\Qsf}.
\end{equation}
Comparing this with the general representation $\bm{\Xsf} = \bm{\Asf}\bm{\Qsf}$ in Lemma~\ref{AQ}, we see that a sensing-optimal choice within this class is to take a deterministic diagonal matrix
\begin{equation}
\bm{A}_{\rm s}^{\star} \triangleq \sqrt{T}\,\bm{D}_{\rm s}^{\star}.
\end{equation}
In the subsequent joint design, we restrict attention to transmit signals of the form
$\bm{\Xsf} = \bm{\Asf}\bm{\Qsf}$ in Lemma~\ref{AQ}, and model the diagonal  
matrix as a \emph{deterministic} design variable within each coherence block. That is, 
we set $\bm{\Asf} = \bm{A}$ a.s. for a generic diagonal matrix 
$\bm{A} \succeq \bm{0}$ (possibly different from $\bm{A}_{\rm s}^{\star}$), while the 
randomness of the input is entirely carried by $\bm{\Qsf}$. 
Furthermore, based on the unitary invariance of both the sensing metric and the 
\ac{id} communication channel, it is without loss of generality to 
work in the \textit{eigen-basis} of $\bm{R}_{\bm{\Hsf}_{\rm s}}$. That is, we can rotate the 
coordinate system by $\bm{U}_{\rm s}$ so that $\bm{R}_{\bm{\Hsf}_{\rm s}}$ becomes 
diagonal, and the design reduces to optimizing the diagonal entries of $\bm{A}$ over $M$ parallel spatial streams. 
\section{Optimal Design and Performance Analysis in High and Low SNR Regimes}
\label{sec:zheng_tse}
We now investigate the optimal signaling strategy for the considered noncoherent ISAC system. Our goal is to characterize the structure of the transmit signal $\bm{\Xsf}$ that balances the conflicting requirements of sensing  and communication.

\subsection{Asymptotic Analysis in the High-SNR Regime}
In this subsection, we consider the high-SNR regime. 
\begin{prop} \label{prop:lower_bound}
In the high-SNR regime, the noncoherent mutual information $I(\bm{\Xsf};\bm{\Ysf}_{\rm c})$ is lower-bounded as
\begin{equation} \label{eq:stat_bound}
\begin{aligned}
I(\bm{\Xsf};\bm{\Ysf}_{\rm c}) &\ge (T - M) \E \left[ \log\operatorname{det}(\bm{\Asf}^2) \right] + \hat{c} + o(1).
\end{aligned}
\end{equation}
\end{prop}

\begin{proof} 
\begin{align}
&I(\bm{\Xsf};\bm{\Ysf}_{\rm c})=h(\bm{\Ysf}_{\rm c})-h(\bm{\Ysf}_{\rm c} \mid \bm{\Xsf})\nonumber\\
 \stackrel{(a)}{=} &h(\bm{\Hsf}_{\rm c}\bm{\Asf}\bm{\Psi}) + (T - M-N_{\rm c}) \mathbb{E} \left[ \log\det(\bm{\Asf}^2) \right] + c_0 + o(1)  \nonumber\\
 \stackrel{(b)}{=} &h(\bm{\Hsf}_{\rm c}\bm{\Asf}\bm{\Psi} \mid \bm{\Asf}, \bm{\Psi})  + I(\bm{\Hsf}_{\rm c}\bm{\Asf}\bm{\Psi} ; \bm{\Asf}, \bm{\Psi}) \nonumber \\
& \quad + (T - M-N_{\rm c}) \mathbb{E} \left[ \log\det(\bm{\Asf}^2) \right] + c_0 + o(1)  \nonumber\\
 \stackrel{(c)}{=}& \mathbb{E} \left[ N_{\rm c} \log\det(\pi e \bm{\Asf}^2) \right] + I(\bm{\Hsf}_{\rm c}\bm{\Asf}\bm{\Psi} ; \bm{\Asf}) + I(\bm{\Hsf}_{\rm c}\bm{\Asf}\bm{\Psi} ; \bm{\Psi} \mid \bm{\Asf})  \nonumber\\
& \quad + (T - M-N_{\rm c}) \mathbb{E} \left[ \log\det(\bm{\Asf}^2) \right] + c_0 + o(1) \nonumber\\
 \stackrel{(d)}{=} &(T - M) \mathbb{E} \left[ \log\det({\bm{\Asf}}^2) \right] + I(\bm{\Hsf}_{\rm c}\bm{\Asf}\bm{\Psi} ; \bm{\Asf})  \nonumber\\
&\quad + I(\bm{\Hsf}_{\rm c}\bm{\Asf}\bm{\Psi} ; \bm{\Psi} \mid \bm{\Asf}) + \hat{c} + o(1)  \nonumber\\
 \stackrel{(e)}{\ge} &(T - M) \mathbb{E} \left[ \log\det(\bm{\Asf}^2) \right] + \hat{c} + o(1)\triangleq I_{\mathrm{lb}}(\bm{\Asf}),\label{eq:derivation}
\end{align}
\RED{where (a) follows from the high-SNR asymptotic expansion in \cite[Eq.~(25) and subsequent equations]{lizhong_zheng_communication_2002} with $c_0$ denoting all terms independent of $\bm{\Asf}$, and $\bm{\Psi} \in \mathbb{C}^{M \times M}$ is an auxiliary \ac{id} unitary matrix independent of $\bm{\Hsf}_{\rm c}$ and $\bm{\Asf}$;}
(b) applies the entropy decomposition;
(c) conditions on $(\bm{\Asf},\bm{\Psi})$ and evaluates the resulting Gaussian entropy, yielding
$h(\bm{\Hsf}_{\rm c}\bm{\Asf}\bm{\Psi} \mid \bm{\Asf},\bm{\Psi}) = \mathbb{E}[N_{\rm c}\log\det(\pi e \bm{\Asf}^2)]$;
(d) collects all $\bm{\Asf}$-dependent terms and absorbs constants into $\hat{c}$;
and (e) follows from the non-negativity of mutual information.
\end{proof}

In particular, if $\bm{\Asf}$ is deterministic, then
$I(\bm{\Hsf}_{\rm c}\bm{\Asf}\bm{\Psi}; \bm{\Asf}) = 0$.
Moreover, in the communication-optimal case where
$\bm{\Asf} = \sqrt{T}\bm{I}_M$ a.s., we have
$I(\bm{\Hsf}_{\rm c}\bm{\Asf}\bm{\Psi}; \bm{\Psi} \mid \bm{\Asf}) = 0$,
since $\sqrt{T}\bm{\Hsf}_{\rm c}\bm{\Psi}$ is independent of $\bm{\Psi}$
under \ac{id} $\bm{\Hsf}_{\rm c}$ (cf.\ Lemma~\ref{id}).
Hence, the lower bound in \eqref{eq:stat_bound} is tight in this case.

To characterize the impact of sensing-driven deviations from uniform power allocation in noncoherent \ac{isac} systems, we transition from the statistical bound in Proposition~\ref{prop:lower_bound} to a structure-aware metric.
\begin{defn} 
For any given  $\bm{A} \succeq \bm{0}$ satisfying $\mathrm{tr}(\bm{A}^2) = MT$, the sensing-induced rate loss is 
\begin{equation}
\begin{aligned}
    \Delta(\bm{A})
    &\triangleq C_{\rm pure} -T^{-1} I_{\mathrm{lb}}(\bm{A}) \\
    &=(1-\frac{M}{T})(M\log T - \log\operatorname{det}(\bm{A}^2)), 
\end{aligned}
\end{equation}
where $C_{\mathrm {pure}}$ denotes the achievable high-SNR rate attained by $\bm{A}^{\star} = \sqrt{T}\bm{I}_M$ (cf.\ Lemma~\ref{cpure}), and $I_{\mathrm{lb}}(\bm{A})$ is the  lower bound obtained by evaluating \eqref{eq:derivation} at the realization $\bm{\Asf} = \bm{A}$.
\end{defn}

\begin{thm}
 The sensing-induced rate loss $\Delta(\bm{A}) \ge 0$ for all feasible $\bm{A}$, with equality if and only if $\bm{A}^2 = T\bm{I}_M$. 
\end{thm}
\begin{proof}
By applying the Arithmetic-Geometric Mean inequality, \(\frac{1}{M}\log\det(\bm{A}^2)= \frac{1}{M} \sum_{i=1}^M \log [\bm{A}]_{ii}^2 \le \log(\frac{1}{M}\sum_{i=1}^M [\bm{A}]_{ii}^2) = \log T\). Thus, $\Delta \ge 0$, with equality holding if and only if all $[\bm{A}]_{ii}$ are equal, i.e., $\bm{A}^2 = T\bm{I}_M$.
\end{proof}

\begin{rem}\label{rem1}
Within the considered noncoherent \ac{isac} framework, a favorable structural compatibility emerges:

    $\bullet$ \textit{Structural Compatibility:} 
    When $\bm{A}$ is deterministic, the semi-unitary constraint $\bm{\Qsf}\bm{\Qsf}^\herm = \bm{I}_M$ ensures that the covariance $\bm{R}_{\bm{\Xsf}} = T^{-1}\bm{A}^2$ is independent of $\bm{\Qsf}$. As a result, the sensing performance depends only on the second-order statistics of the transmit signal and is decoupled from the specific realizations of the random communication symbols.
    
    $\bullet$ \textit{Tradeoff Vanishing:} 
    While noncoherent communication favors equal-power signaling ($\bm{\Xsf} = \sqrt{T}\bm{\Qsf}$) in the high-SNR regime, the sensing-optimal design reduces to the same structure when the sensing channel is isotropic ($\bm{R}_{\bm{\Hsf}_{\rm s}} \propto \bm{I}_M$) (cf.\ Proposition~\ref{sensing_optimal}).
Consequently, under isotropic sensing environments and high-SNR asymptotics, the design objectives for \ac{snc}  become structurally compatible, implying that the tradeoff vanishes under the adopted performance metrics.
\end{rem}
\vspace{-0.3em}
Building on the sensing-induced rate loss metric defined above and the effective communication input representation developed in Subsection~\ref{subsec:eff_comm_input}, the joint \ac{snc} design can be formulated as a tradeoff between the rate penalty and the sensing error. Since the problem is fully parameterized by the diagonal power allocation matrix $\hat{\bm{A}}=\bm{A}^2= \mathrm{diag}(\hat{a}_1,\ldots,\hat{a}_M)$, this leads to the following optimization problem:
\begin{subequations}\label{OPTA1}
\begin{align}
\min_{\hat{\bm{A}}} \quad
& F(\hat{\bm{A}})
\triangleq \alpha\,\hat{\Delta}(\hat{\bm{A}})
+ (1-\alpha)\,S(\hat{\bm{A}}) \\
\mathrm{s.t.} \quad
& \mathrm{tr}(\hat{\bm{A}}) = MT, \,
\hat{\bm{A}} \succeq \bm{0},
\end{align}
\end{subequations}
\RED{where $\alpha \in [0,1]$ is the weighting factor that implicitly incorporates the necessary normalization to commensurate the disparate orders of magnitude of the two metrics,}
$\hat{\Delta}(\hat{\bm{A}})$ denotes the loss
$\Delta(\bm{A})$ expressed in terms of
$\hat{\bm{A}} = \bm{A}^2$, and
$S(\hat{\bm{A}})
= \sum_{i=1}^M (\lambda_i^{-1} + \kappa \hat{a}_i)^{-1}$
denotes the sensing metric.

\begin{prop}\label{prop:opta_convex}
Problem \eqref{OPTA1} is a convex optimization problem with respect to $\hat{\bm{A}}$.
\end{prop}
\begin{proof}
(sketch) The objective is convex since $-\log\det(\hat{\bm{A}})$ is convex on the positive semidefinite cone and
$S(\hat{\bm{A}})$ is a sum of convex functions in $\hat{a}_i$, while the constraint is affine.
For a detailed proof, please refer to Appendix~\ref{app:opta_convex}.
\end{proof}
\vspace{-0.3em}
The optimal power allocation $\hat{\bm{A}}^\star$ is computed via a \ac{pga}. The gradient $\bm{G} = \nabla F(\hat{\bm{A}})$ is a diagonal matrix with entries:
\begin{equation}\label{pga}
    [\bm{G}]_{ii} = -\frac{\alpha (1-\frac{M}{T})}{\hat{a}_i} - \frac{(1-\alpha) \kappa}{(\lambda_{i}^{-1} + \kappa \hat{a}_i)^{2}}, \, \forall i=1,\dots,M.
\end{equation}
In each iteration, the variable is updated via gradient descent followed by a projection $\mathcal{P}_{\mathcal{S}}$ onto the simplex constraint set $\{\hat{\bm{A}} \succeq \bm{0} \mid \mathrm{tr}(\hat{\bm{A}}) = MT\}$. Upon convergence, the optimal solution is recovered as $\bm{A}^{\star} = (\hat{\bm{A}}^{\star})^{1/2}$.

\begin{figure}[t!]
\vspace{0.05in}
  \centering 
  \begin{minipage}{0.48\linewidth}
    \centering 
    \includegraphics[width=\linewidth]{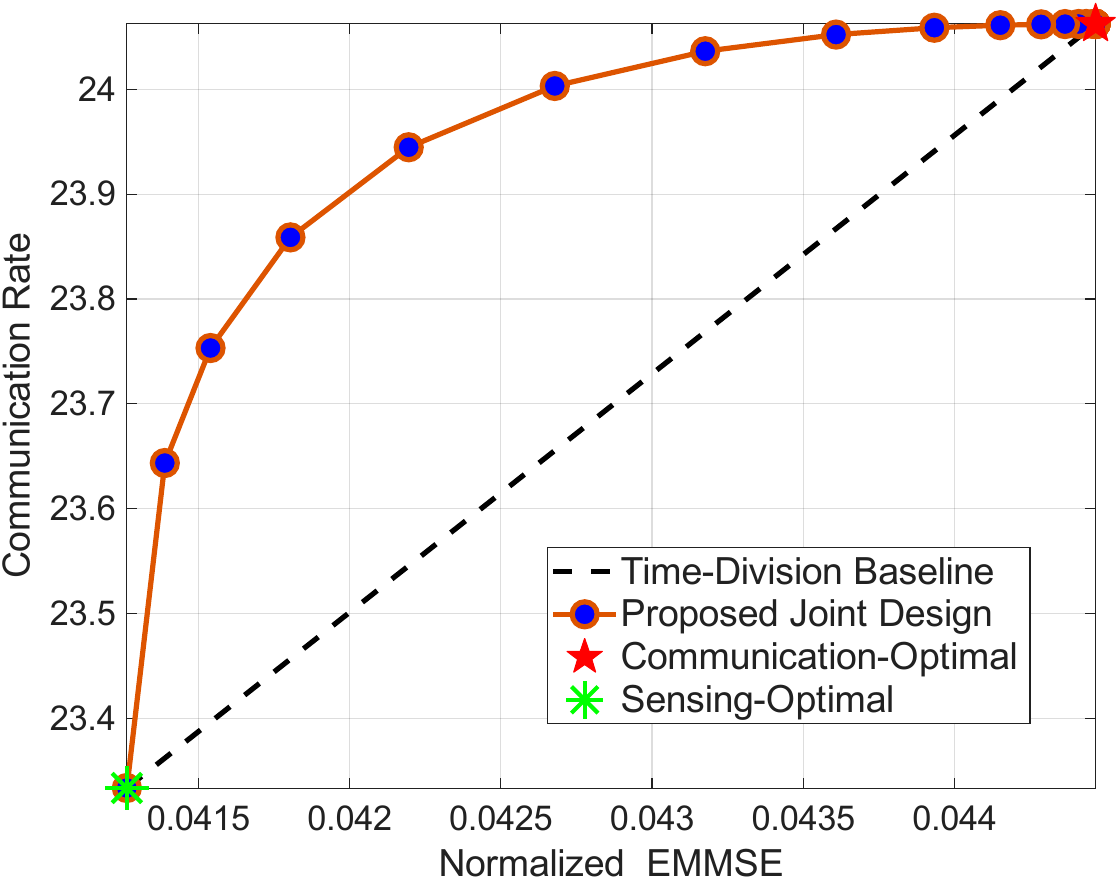}   
    \caption{The tradeoff between communication rate and normalized EMMSE.}
    \label{fig:tradeoff} 
  \end{minipage}
  \hfill
  \begin{minipage}{0.48\linewidth}
    \centering
    \includegraphics[width=\linewidth]{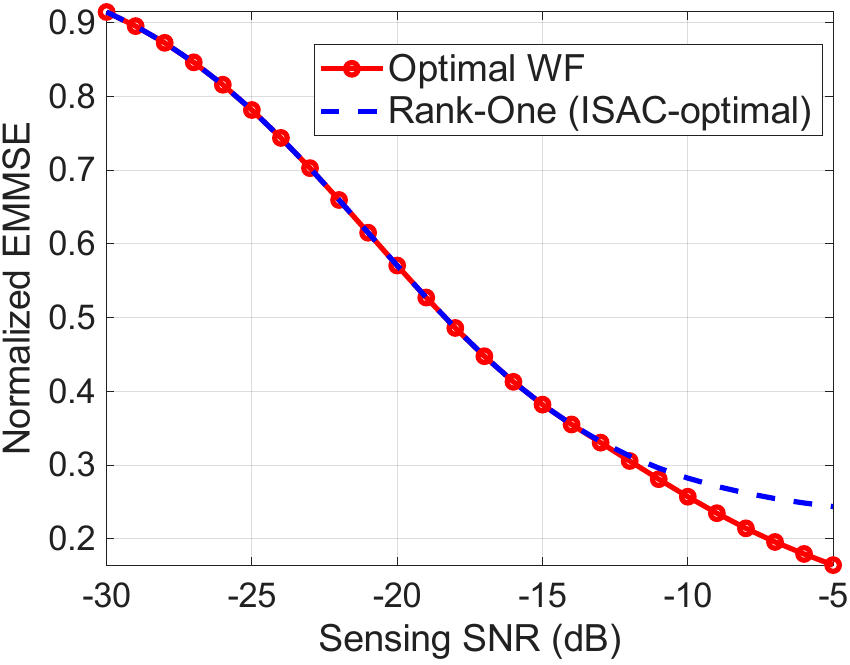} 
    \caption{Low-SNR sensing performance under different  power allocation strategies.}
    \label{fig:low_snr_sensing}
  \end{minipage}
  \vspace{-6mm}
\end{figure}

\textit{Simulation Results:}
We evaluate the proposed framework with $M = N_{\rm c} = N_{\rm s} = 8$, $T = 20$, $\SNR_{\rm c} = 20$~dB, and $\SNR_{\rm s} = 8$~dB, \RED{which accounts for the fact that radar sensing typically suffers from severe two-way path loss compared to the one-way communication link.}
According to Remark~\ref{rem1}, when 
the sensing channel is isotropic, the \ac{snc} tradeoff
vanishes.
To explicitly demonstrate the non-trivial tradeoff behavior, we therefore
consider a non-isotropic sensing channel model with
$[\bm{R}_{\bm{\Hsf}_{\rm s}}]_{ij} = N_{\rm s}\rho^{|i-j|}$, where $\rho = 0.9$.
We plot the normalized EMMSE, defined as
$J_{\EMMSE} / \mathrm{tr}(\bm{R}_{\bm{\Hsf}_{\rm s}})$,
versus the communication rate
$T^{-1} I_{\mathrm{lb}} = C_{\rm pure} - \Delta$.
As illustrated in Fig.~\ref{fig:tradeoff}, there exists a clear Pareto frontier representing the tradeoff between \ac{snc}. Specifically, when $\alpha \to 1$, the system prioritizes the communication objective, driving the rate $T^{-1}I_{\mathrm{lb}}$ to its maximum value $C_{\rm pure}$, albeit at the cost of a significantly increased sensing metric. Conversely, as $\alpha$ decreases towards $0$, the optimization shifts its focus toward sensing, resulting in a minimized EMMSE as the power allocation $\hat{\bm{A}}$ is optimized to align with the eigenvalues of $\bm{R}_{\bm \Hsf_{\rm s}}$.

\vspace{-0.3em}
\subsection{Asymptotic Analysis in the Low-SNR Regime}
\label{sec:low_snr}


\subsubsection{Power-Limited Communication and Flash Signaling}
Under the noncoherent block-fading model in Sec.~\ref{sec:model}, the low-SNR capacity
admits the first-order expansion \cite{lizhong_zheng_communication_2002,verdu_spectral_2002}
\begin{equation}
   C(\SNR_{\rm c}) = N_{\rm c}\log_2(e)\,\SNR_{\rm c} + o(\SNR_{\rm c}),
   \label{eq:low_snr_slope}
\end{equation}
which coincides with the coherent benchmark to first order. Hence, to a first order, there is no capacity penalty for not knowing the channel at the receiver in the low-SNR
regime, unlike in the high-SNR regime, where the noncoherent channel is degree of freedom
limited.

Moreover, the first-order gain from multiple antennas comes from the increase in total
received power due to multiple receive antennas (the linear factor $N_{\rm c}$), whereas
multiple transmit antennas provide no first-order improvement.
The first-order term in \eqref{eq:low_snr_slope} can be asymptotically achieved by
\emph{flash (temporally peaky) signaling}: allocating essentially all transmit energy in each
coherence interval to a single symbol time and a single transmit spatial dimension (e.g.,
one transmit antenna/one beamforming direction), while the receiver adds up the received
energies from the $N_{\rm c}$ antennas noncoherently.

\subsubsection{Low-SNR Sensing Behavior}
As $\SNR_{\rm s}\to 0$ (i.e., $\kappa \to 0$), the $J_{\EMMSE}$ in \eqref{def:emmse} can be approximated as
\begin{align}
   J_{\EMMSE}&\stackrel{(a)}{=} \mathbb{E}_{\bm{\Xsf}} \left[ 
        \tr \left( 
            \bm{R}_{\bm{\Hsf}_{\rm s}} 
            - \kappa \bm{R}_{\bm{\Hsf}_{\rm s}} 
                \bm{\Xsf}\bm{\Xsf}^\mathsf{H} 
                \bm{R}_{\bm{\Hsf}_{\rm s}} 
        \right) 
    \right] 
    + o(\kappa) \nonumber \\
    &= \tr(\bm{R}_{\bm{\Hsf}_{\rm s}}) 
    - \kappa \tr \big( \bm{R}_{\bm{\Hsf}_{\rm s}}^2 \bm{\varSigma} \big) 
    + o(\kappa), \label{eq:low_snr_expansion_final}
\end{align}
where $\bm\varSigma=\E[\bm \Xsf \bm \Xsf^\herm]$, and (a) follows from the first-order matrix inverse expansion
\(
    (\bm{A}^{-1} + \kappa \bm{B})^{-1} 
    = \bm{A} - \kappa \bm{A}\bm{B}\bm{A} + o(\kappa)
\)
for small $\kappa$ \cite{horn2012matrix}. 
Minimizing $J_{\EMMSE}$ is  equivalent to solving
\begin{equation}
    \max_{\bm{\varSigma} \succeq \bm{0}}\, \tr(\bm{R}_{\bm{\Hsf}_{\rm s}}^2 \bm{\varSigma}), \quad \text{s.t. } \tr(\bm{\varSigma}) = MT. \label{eq:low_snr_opt_prob}
\end{equation}

\begin{thm}
\label{thm:low_snr_sensing}
In the low-SNR regime, a sensing-optimal $\bm{\varSigma}^\star$ is rank-one and given by $\bm{\varSigma}^\star = MT \bm{u}_{\max} \bm{u}_{\max}^\mathsf{H}$, where $\bm{u}_{\max}$ is the dominant eigenvector of $\bm{R}_{\bm{\Hsf}_{\rm s}}$.
\end{thm}

\begin{proof}
By von Neumann's trace inequality \cite{horn2012matrix}, 
$
\tr(\bm{R}_{\bm{\Hsf}_{\rm s}}^2 \bm{\varSigma}) \le \sum_{i=1}^M \lambda_i(\bm{R}_{\bm{\Hsf}_{\rm s}}^2) \lambda_i(\bm{\varSigma}) \le MT \cdot \lambda_{\max}^2(\bm{R}_{\bm{\Hsf}_{\rm s}}),
$
where $\lambda_i(\cdot)$ denotes eigenvalues in nonincreasing order. The upper bound is achieved by the rank-one $\bm{\varSigma}^\star = MT \bm{u}_{\max} \bm{u}_{\max}^\mathsf{H}$.
\end{proof}
\begin{cor}[First-Order Optimal ISAC Signal Structure]
\label{cor:signal_structure}
A first-order optimal ISAC waveform can be chosen to be rank-one in space and flash in time, e.g.,
\begin{equation}\label{signal_structure}
    \bm{\Xsf}_{\rm opt} = \sqrt{MT}\,\bm{u}_{\max}\,\xi\,\bm{e}_1^{\Tsf},
\end{equation}
where $\bm{e}_1\in\mathbb{C}^{T}$ is the first  basis vector and $\xi\in\mathbb{C}$ is a random scalar
satisfying $\E[|\xi|^2]=1$ with a peaky/flash distribution (zero most of the time and very large with very small probability).
\end{cor}
\begin{rem}[Low-SNR Tradeoff Vanishing]
Corollary~\ref{cor:signal_structure} satisfies $\E[\bm{\Xsf}_{\rm opt}\bm{\Xsf}_{\rm opt}^{\herm}]
= MT\,\bm{u}_{\max}\bm{u}_{\max}^{\herm}$ and is therefore sensing-optimal to first order by
\eqref{eq:low_snr_expansion_final}. Meanwhile, it follows that allocating
all transmit energy in each coherence interval to a single symbol and a single transmit spatial dimension is sufficient
to asymptotically achieve the first-order term in \eqref{eq:low_snr_slope}. Since the channel is isotropic,
choosing the direction as $\bm{u}_{\max}$ incurs no first-order communication loss. Hence, the \ac{snc}
tradeoff vanishes asymptotically to first order at low SNR.
\end{rem}

\textit{Simulation Results:}
We evaluate the normalized EMMSE in the low-SNR regime, to illustrate the accuracy of the  asymptotic analysis.
The system parameters are set to $M = N_{\rm s} = 8$ and $T = 20$, 
with the same $\bm{R}_{\bm{\Hsf}_{\rm s}}$ as in Fig.~\ref{fig:tradeoff}.
Specifically, Fig.~\ref{fig:low_snr_sensing} compares two spatial power allocation strategies:
(i) the sensing-optimal water-filling solution in Proposition~\ref{sensing_optimal}, and
(ii) the rank-one signaling strategy predicted by  Corollary~\ref{cor:signal_structure}.
As $\SNR_{\rm s}$ decreases, the performance gap between the water-filling solution and the rank-one strategy gradually vanishes.
This convergence empirically confirms that the rank-one spatial structure becomes asymptotically optimal in the low-SNR regime.
The communication performance is not explicitly plotted because, to first order as $\SNR_{\rm c}\to 0$, the achievable rate is governed by the total transmit energy.
Under the same block-energy constraint and the same temporally peaky/flash signaling pattern, different spatial covariances do not affect the first-order term in the low-SNR expansion. 

\section{Conclusion and Discussion}
This paper provided a new perspective on the fundamental limits of noncoherent \ac{isac} systems. Contrary to the conventional intuition that communication and sensing are inherently conflicting, our analysis revealed that their structural alignment is highly \ac{snr}-dependent. Specifically, the identified vanishing tradeoff phenomenon---occurring not only in the low-\ac{snr} regime but also under specific high-\ac{snr} configurations---suggests that joint signaling can achieve near-optimal performance for both functionalities simultaneously. This convergence implies that the perceived mismatch between the \ac{ustm} signaling and sensing-optimal power allocation is a flexible constraint rather than a fundamental barrier. 
\section*{Acknowledgement}

 The work of H. Yang and K. Wan was funded by NSFC-12141107 and Wuhan ``Chen Guang'' Program under Grant 2024040801020211, and the work of G. Caire was supported by the Gottfried Wilhelm Leibniz-Preis 2021 of the German Science Foundation (DFG).
\appendices
\section{Proof of Proposition~\ref{sensing_optimal}}
\label{app:sensing_optimal}
Recalling the definition in \eqref{def:emmse}, let
\[
f(\bm{\Rsf}_{\bm{\Xsf}})=
\tr\left((\bm{R}_{\bm{\Hsf}_{\rm s}}^{-1}
+\kappa T\bm{\Rsf}_{\bm{\Xsf}})^{-1}\right).
\]
The mapping $\bm{S}\mapsto
(\bm{R}_{\bm{\Hsf}_{\rm s}}^{-1}+\kappa T\bm{S})^{-1}$ is matrix convex on the positive semidefinite cone, and the trace operator is linear; hence $f(\cdot)$ is convex in $\bm{\Rsf}_{\bm{\Xsf}}\succeq\bm{0}$.
By Jensen's inequality,
\[
\E[f(\bm{\Rsf}_{\bm{\Xsf}})] \ge
f(\E[\bm{\Rsf}_{\bm{\Xsf}}])
\triangleq f(\widetilde{\bm{R}}_{\bm{\Xsf}}),
\]
with equality if $\bm{\Rsf}_{\bm{\Xsf}}$ is deterministic, i.e.,
$\bm{\Rsf}_{\bm{\Xsf}}=\widetilde{\bm{R}}_{\bm{\Xsf}}$ almost surely.
This proves Item 1).

For a deterministic $\bm{\Rsf}_{\bm{\Xsf}}$, the objective
$f(\bm{\Rsf}_{\bm{\Xsf}})$ is minimized when its eigenvectors align with those of $\bm{R}_{\bm{\Hsf}_{\rm s}}^{-1}$, equivalently with those of $\bm{R}_{\bm{\Hsf}_{\rm s}}$, by the standard trace/eigenvalue alignment argument. This yields Item 2).
Substituting the eigen-decompositions into $f(\cdot)$ reduces the problem to the scalar convex optimization
\[
\min_{\{d_i\ge 0\}}
\sum_{i=1}^{M}
\left(\lambda_i^{-1}+\kappa T d_i^2\right)^{-1},
\quad
\text{s.t.}\quad \sum_{i=1}^{M}d_i^2=M.
\]
The KKT conditions give
\[
(d_i^\star)^2
=\frac{1}{\kappa T}\left(\nu-\frac{1}{\lambda_i}\right)^+,
\quad i=1,\ldots,M,
\]
where $\nu$ is chosen such that $\sum_{i=1}^{M}(d_i^\star)^2=M$.
This proves Item 3). Item 4) follows directly when
$\bm{R}_{\bm{\Hsf}_{\rm s}}\propto\bm{I}_M$, in which case the symmetry of the scalar problem gives $d_i^\star=1$ for all $i=1,\ldots,M$.

\section{Proof of Proposition~\ref{prop:opta_convex}}
\label{app:opta_convex}
Using the definition of the sensing-induced rate loss, the term
$\hat{\Delta}(\hat{\bm{A}})$ can be written, up to constants independent of $\hat{\bm{A}}$, as
\[
\hat{\Delta}(\hat{\bm{A}})
=-\left(1-\frac{M}{T}\right)\log\det(\hat{\bm{A}})+\mathrm{const.}
\]
Since $-\log\det(\hat{\bm{A}})$ is convex on the positive definite cone, $\hat{\Delta}(\hat{\bm{A}})$ is convex on its domain.
The sensing objective is separable:
\[
S(\hat{\bm{A}})
=\sum_{i=1}^{M}s_i(\hat{a}_i),
\qquad
s_i(\hat{a}_i)=\left(\lambda_i^{-1}+\kappa\hat{a}_i\right)^{-1}.
\]
For $\hat{a}_i\ge 0$, its second derivative is
\[
s_i''(\hat{a}_i)
=\frac{2\kappa^2}{(\lambda_i^{-1}+\kappa\hat{a}_i)^3}\ge 0,
\]
so each $s_i(\cdot)$ is convex and therefore $S(\hat{\bm{A}})$ is convex.
For any $\alpha\in[0,1]$, $F(\hat{\bm{A}})=\alpha\hat{\Delta}(\hat{\bm{A}})+(1-\alpha)S(\hat{\bm{A}})$ is a nonnegative weighted sum of convex functions.
Finally, the feasible set
\[
\{\hat{\bm{A}}\succeq\bm{0}\mid \tr(\hat{\bm{A}})=MT\}
\]
is convex because it is the intersection of the positive semidefinite cone and an affine hyperplane.
Thus, Problem~\eqref{OPTA1} is convex with respect to $\hat{\bm{A}}$.

 \bibliographystyle{IEEEtran}
\bibliography{IEEEabrv,isit}

\end{document}